\documentclass{article}
 
\usepackage[margin=1in]{geometry}
\usepackage{tikz}
\usepackage{amsmath}
\usepackage{amssymb}
\usepackage{amsthm}
\usepackage[boxed]{algorithm2e}
\usepackage[inline]{enumitem}
\usetikzlibrary{decorations.pathreplacing}

\title{A note on the integrality gap of the configuration LP for restricted Santa Claus\footnote{Research
was supported by German Research Foundation (DFG) project JA 612/15-2}}

\author{Klaus Jansen \and Lars Rohwedder}
\date{Christian-Albrechts-Universit\"at, Kiel, Germany \\
  \texttt{\{kj, lro\}@informatik.uni-kiel.de}}

\newtheorem{theorem}{Theorem}

\newtheorem{claim}{Claim}

\newcommand{\players}{\mathcal P}
\newcommand{\resources}{\mathcal R}
\newcommand{\OPT}{\mathrm{OPT}}

\begin{document}

\maketitle
\begin{abstract}
  In the restricted Santa Claus problem we are given resources $\resources$ and players $\players$.
  Every resource $j\in\resources$ has a value $v_j$ and
  every player $i$ desires a set $\resources(i)$ of resources.
  We are interested in distributing the resources to players that desire them.
  The quality of a solution is measured by the least happy player, i.e., the lowest sum of 
  resource values. This value should be maximized.
  The local search algorithm by Asadpour et al.~\cite{DBLP:journals/talg/AsadpourFS12} and its 
  connection to the configuration LP has proved itself to be a very influential technique
  for this and related problems.

  In the original proof, a local search was used to obtain a bound of $4$
  for the ratio of the fractional
  to the integral optimum of the configuration LP (integrality gap). This bound is
  non-constructive since the local search has not been shown to terminate in polynomial time.
  On the negative side, the worst instance known has an integrality gap of $2$.

  Although much progress was made in this area, neither bound has been improved since.
  We present a better analysis that shows the integrality gap is not
  worse than~$3 + 5/6 \approx 3.8333$.
\end{abstract}
\section{Introduction}
A generalization of the problem we consider goes back to Bansal and Srividenko~\cite{DBLP:conf/stoc/BansalS06}.
In the Santa Claus problem there are players $\players$ and resources $\resources$. Every
resource $j$ has a value $v_{ij} \ge 0$ for player $i$. The goal is to find an assignment $\sigma : \resources\rightarrow\players$ such that $\min_{i\in\players}\sum_{j\in\sigma^{-1}(i)} v_{ij}$ is maximized.

In the restricted variant, we consider only values $v_{ij} \in\{0, v_j\}$ where $v_j > 0$ is a value
depending only on the resource. This can also be seen as each player desiring a subset $\resources(i)$
of resources which have a value of $v_j$ for him, whereas other resources cannot be assigned to him.

For the restricted Santa Claus problem there exists a strong LP relaxation, the configuration LP.
The proof that this has a small integrality gap (see \cite{DBLP:journals/talg/AsadpourFS12}) is not trivial.
It works by defining an
exponential time local search algorithm which is guaranteed to return an integral solution of
value not much less than the fractional optimum.
This technique has since been used in other problems, like the minimization of the makespan~\cite{DBLP:journals/siamcomp/Svensson12, DBLP:conf/soda/JansenR17}.
Significant research has also gone into making the proof constructive~\cite{PolacekS16,DBLP:journals/talg/AnnamalaiKS17,DBLP:conf/ipco/JansenR17}.
Yet, no improvement of the bound of $4$ on the integrality gap has been found. We show that
the original analysis is not tight and can be improved to $3 + 5/6\approx 3.8333$.
\subsection{Configuration LP}
the configuration LP is an exponential size LP relaxation, but it
can be approximated in polynomial time with a rate of $(1 + \epsilon$)
for every $\epsilon > 0$~\cite{DBLP:conf/stoc/BansalS06}.
For every player $i$ and every value $\tau$ let
\begin{equation*}
  \mathcal C(i, \tau) = \{ S \subseteq \resources(i) : v(S) \ge \tau\} .
\end{equation*}
These are the configurations for player $i$ and value $\tau$. They are a selection of
resources that have value at least $\tau$ and are desired by player $i$.
The optimum $\OPT^*$ of the configuration LP is the highest $\tau$ such that the following linear
program is feasible.
\\[1em]
\fbox{
\begin{minipage}{\textwidth}
Primal of the configuration LP for restricted {\sc Santa Claus}
\begin{align*}
  \sum_{C\in\mathcal C(i, \tau)} x_{i, C} &\ge 1 & \forall i\in\players \\
  \sum_{i\in\players}\sum_{C\in\mathcal C(i,\tau) : j\in C} x_{i, C} &\le 1 & \forall j\in\resources \\
  x_{i, C} &\ge 0
\end{align*}
\end{minipage}}
\\[1em]
\fbox{
\begin{minipage}{\textwidth}
Dual of the configuration LP for restricted {\sc Santa Claus}
\begin{align*}
  \max \sum_{i\in\players} y_i &- \sum_{j\in\resources} z_j \\
  \sum_{j\in C} z_j &\ge y_i &\forall i\in\players, C\in\mathcal C(i, \tau) \\
  y_i, z_j &\ge 0
\end{align*}
\end{minipage}}
\\[1em]
We derive the following condition from duality:
\begin{theorem}\label{theorem-condition-sc}
  Let $y\in \mathbb R_{\ge 0}^\players$ and $z\in \mathbb R_{\ge 0}^\resources$ such
  that $\sum_{i\in\players} y_i > \sum_{j\in\resources} z_j$ and
  for every $i\in\players$ and $C\in\mathcal C(i, \tau)$ it holds
  that $\sum_{j\in C} z_j \ge y_i$, then
  $\OPT^* < \tau$.
\end{theorem}
It is easy to see that if such a solution $y, z$ exists, then every component can be scaled
by a constant to obtain a feasible solution greater than any given value. Hence, the 
dual must be unbounded and therefore the primal must be infeasible.
\section{Algorithm}
We consider the local search algorithm from~\cite{DBLP:journals/talg/AsadpourFS12}.
It is the same algorithm with a slightly different presentation that is
inspired by~\cite{PolacekS16}.
Throughout this section we will denote by $\alpha = 3 + 5/6$ the bound on integrality
we want to prove.

We model our problem as a hypergraph matching problem:
There are vertices for all players and all resources and the hyperedges $\mathcal H$ each consist
of exactly one player $i$ and a set of resources $C\subseteq \resources(i)$ where
$v(C) \ge \OPT^* / \alpha$.
However, we restrict $\mathcal H$ to edges that are minimal,
that is to say $v(C') < \OPT^* / \alpha$ for all
$C'\subset C$.
It is easy to see that a matching (a set of non-overlapping edges) such that every
player is in one matching edge corresponds to a solution of value $\OPT^*/\alpha$.

For a set of edges $F$ we write $F_\players$ as the set of players in these edges
and $F_\resources$ as the resources in the edges.
The algorithm maintains a partial matching $M$ and extends it one player at a time. After $|\players|$ many calls to the algorithm
the desired matching is found.
Two types of edges play a crucial role in the algorithm:
An ordered list $A = \{e^A_1,\dotsc, e^A_\ell\} \subseteq \mathcal H\setminus M$ (the addable edges)
and sets $B_M(e^A_1)$, $B_M(e^A_2)$, \dots, $B_M(e^A_\ell) \subseteq M$ (the blocking edges for each addable edge).
An addable edge is a edge that the algorithm hopes to add to $M$ - either to cover the new player or to free the player of a blocking edge.
A blocking edge is an edge in $M$ that conflicts with an addable edge, i.e., that has a non-empty overlap with an addable edge.
For each addable edge $e^A_k$ we define the blocking edges $B_M(e^A_k)$ 
as $\{e'\in M : e'_\resources \cap (e^A_k)_\resources \neq \emptyset\}$.
From the definition of the algorithm it will be clear that
$B_M(e^A_k) \cap B_M(e^A_{k'}) = \emptyset$ for $k\neq k'$.
We write $B_M(A)$ for $\bigcup_{e\in A} B_M(e)$.

\subsection{Detailed description of the algorithm}
In each iteration the algorithm first adds a new addable edge
that does not overlap in resources with any existing addable
or blocking edge. Then it consecutively swaps addable
edges that are not blocked for the blocking edge they
are supposed to free.
Also, addable/blocking edges added at a later time are removed,
since they might be obsolete.
The swap does not create new blocking edges,
since the new matching edge does not overlap with addable edges.
Also, by adding only addable edges that do not overlap with
the previous addable/blocking edges, the resources
of addable edges are disjoint and the resources of
each blocking edge overlap only with one addable edge.
The blocking edges $B_M(e^A_k)$ and $B_M(e^A_{k'})$ for
two addable edges with $k < k'$ must be disjoint, because
otherwise $e^A_{k'}$ could not have been added in the first place.
\vspace{1em}
\begin{algorithm}[H]
  \SetKwInOut{KwInput}{Input}
  \SetKwFor{Loop}{loop}{}{end}
  \KwInput{Partial matching $M$ and unmatched player $i_0$}
  \KwResult{Partial matching $M'$ with $M'_\players = M_\players\cup \{i_0\}$}
$\ell \gets 0$ \;
  \Loop{}{
  $\ell \gets \ell + 1$ \;
  let $e^A_\ell\in \mathcal H\setminus M$ with
   $(e^A_\ell)_\resources \cap (A\cup B_M(A))_\resources = \emptyset$ and
    $(e^A_\ell)_\players \in B_M(A)_\players \cup \{i_0\}$ \;
    \tcp{existence of $e^A_\ell$ is proved in the analysis}
  $A \gets A\cup\{e^A_{\ell}\}$ ; \tcp{$e^A_\ell$ is added as the last addable edge}
  \While{$B_M(e^A_\ell) = \emptyset$}{
     \eIf{there is an edge $e'\in B_M(A)$ with $e'_\players = (e^A_\ell)_\players$ \tcp{unambiguous since $M$ is matching}} {
           let $e^A_k\in A$ such that $e'\in B_M(e^A_k)$ ; \tcp{unambiguous}
           $M \gets M\setminus \{e'\}\cup\{e^A_\ell\}$ ; \tcp{swap $e'$ for $e^A_\ell$ \hfill}
           $A \gets \{e^A_1,\dotsc, e^A_k\}$; $\ell\leftarrow k$ ; \tcp{Forget $e^A_{k+1},\dotsc, e^A_\ell$ \hfill}
    }{
          $M \gets M\cup \{e^A_\ell\}$ ;
          \tcp{$(e^A_\ell)_\players = i_0$}
          \Return $M$ \;
    }
  }
  \caption{Local search for restricted {\sc Santa Claus}}
  }
\end{algorithm}
\section{Analysis}
\begin{theorem}[\cite{DBLP:journals/talg/AsadpourFS12}]
  The algorithm terminates after at most $2^{|M| - 1}$ many iterations of the main loop.
\end{theorem}
\begin{proof}
  Consider the signature vector
  \begin{equation*}
    s(A) = \{|B_M(e^A_1)|,|B_M(e^A_2)|,\dotsc,|B_M(e^A_\ell)|, \infty\} .
  \end{equation*}
  This vector decreases after every iteration of the main loop:
  If the inner loop is never executed, then the last component
  is replaced by a finite value. Hence, assume the inner while
  loop is executed at least once. Let $\ell'$ be the cardinality
  of $A$ after the last execution of the inner loop.
  Then $|B(e^A_{\ell'})|$ has decreased by the swap operation.
  It follows that the signature vectors in each iteration of
  the main loop are pairwise different.

  The number of signature vectors can be trivially bounded 
  by $|M|^{|M|}$.
  A clever idea from~\cite{DBLP:journals/talg/AsadpourFS12} 
  even gives a bound of $2^{|M| - 1}$:
  We have that $\sum_{k=1}^\ell |B_M(e^A_k)| = |B_M(A)|\le |M|$
  and $|B_M(e^A_k)| \ge 1$ for all $k$.
  There is an bijection between signature vectors and possibilities
  of placing separators on a line of $|M|$ elements.
  $|B_M(e^A_k)|$ is the number of elements between the $(k-1)$-th
  and $k$-th separator. The number of possibilities of placing
  separators between the $|M|$ elements is the number of
  subsets of $|M| - 1$ elements, i.e. $2^{|M| - 1}$.
\end{proof}
Clearly the inner loop also terminates after finitely many iterations, since in
each iteration $\ell$ is decreased.
\begin{theorem}
  If the configuration LP is feasible and there will an edge that can be added to $A$ as long as $i_0$ is not matched.
\end{theorem}
\begin{proof}
In the proof we use the constant $\beta = 1 + 8/15 \approx 1.53333$, that has been chosen
so as to minimize $\alpha$.
Assume toward contradiction that edge remains that can be added to $A$,
but $i_0$ is not covered.
In the remainder of the proof we will write $B$ instead of $B_M(A)$
and $B(e)$ instead of $B_M(e)$,
since $M$ and $A$ are constant throughout the proof.
Define
\begin{equation*}
  y_i = \begin{cases}
    1 &\text{if $i\in B_\players\cup \{i_0\}$},\\
    0 &\text{otherwise}.\\
  \end{cases}
\end{equation*}
\begin{equation*}
  z_j = \begin{cases}
    1 &\text{if } v_j \ge \OPT^*/\alpha \text{ and } j \in A_\resources \cup B_\resources,\\
    \min\{1/3, \beta \cdot v_j / \OPT^*\}  &\text{if } v_j < \OPT^*/\alpha \text{ and } j \in A_\resources \cup B_\resources,\\
    0 &\text{otherwise}.\\
  \end{cases}
\end{equation*}
We refer to the resources $j$ where $v_j \ge \OPT^*/\alpha$ as
fat resources and to others as thin resources. Note that by
minimality of edges in $\mathcal H$, each edge containing
a fat resource does not contain any other resources.
We call these the fat edges. Likewise, edges that contain only thin
resources are referred to as thin edges.
\begin{claim}
  \label{claim-feasibility}
  $(y, z)$ is a feasible solution for the dual of the configuration LP.
\end{claim}
\begin{claim}
  \label{claim-negativity}
  $(y, z)$ has a negative objective value, that is to say $\sum_{j\in\resources} z_j < \sum_{i\in\players} y_i$.
\end{claim}
By Theorem~\ref{theorem-condition-sc} this implies that the configuration LP is infeasible for $\OPT^*$.
A contradiction.
\end{proof}
\begin{proof}[Proof of Claim~\ref{claim-feasibility}]
Let $i\in\players$ and $C\in \mathcal C(i, \OPT^*)$. We need to show that $y_i \le z(C)$.
If $y_i = 0$ or $C$ contains a fat resource, this is trivial.
Hence, assume w.l.o.g. that $C$ consists solely of thin resources and $y_i = 1$.

Since no addable edge for $i$ remains,
$v(C \setminus (A_\resources \cup B_\resources)) < \OPT^*/\alpha$.
Let $S\subseteq C$ be the resources $j\in C$ which have $z_j = 1/3$.
\begin{description}
  \item[Case 1: $|S| = 3$.] Then $z(C) \ge z(S) \ge 3 \cdot 1/3 = 1$.
  \item[Case 2: $|S| \le 2$.] 
Define $C' := (C \cap (A_\resources \cup B_\resources)) \setminus S$.
Then
\begin{equation*}
  v(C') > v(C) - v(C \setminus (A_\resources \cup B_\resources)) - v(S) \ge \OPT^* - \OPT^*/\alpha - v(S) .
\end{equation*}
Therefore,
\begin{multline*}
  z(C) \ge 1/3 \cdot |S| + \beta / \OPT^* \cdot v(C') > 1/3 |S| + \beta (1 - 1/\alpha - v(S) / \OPT^*) \\
  \ge 1/3 |S| + \beta (1 - (|S| + 1) /\alpha) =: (*) .
\end{multline*}
Since $\beta / \alpha = 0.4 > 1/3$, the coefficient of $|S|$ in $(*)$ is negative
and thus we can substitute $|S|$ for its upper bound, i.e., $2$.
By inserting the values of $\alpha$ and $\beta$ we get,
\begin{equation*}
  (*) \ge 2/3 + \beta (1 - 3/\alpha) = 1 . \qedhere
\end{equation*}
\end{description}
\end{proof}
\begin{proof}[Proof of Claim~\ref{claim-negativity}]
We write in the following $F^f$ ($F^t$) for the fat edges (thin edges, respectively) in a set of edges $F$.
First note that every fat edge with positive $z$ value must be in a fat blocking edge and therefore
\begin{equation*}
  \sum_{j\in \mathcal R^f} z_j \le |B^f| .
\end{equation*}
Now consider thin addable edges.
Since every addable edge is blocked, $|B(e)| \ge 1$ for every $e\in A^t$.
We now proceed to show that for every $e\in A^t$
\begin{equation*}
  z(e_\resources\cup B(e)_\resources) \le |B(e)|.
\end{equation*}
Note that for every thin resource $j$, we have $z_j \le \beta / \OPT^* \cdot v_j$. By minimality of edges in $\mathcal H$, it holds that
$v(e_\resources) \le 2 \OPT^* / \alpha$ (each element
in $e_\resources$ has value at most $\OPT^* / \alpha$).
Also $v(e'_\resources \setminus e_\resources) \le \OPT^*/\alpha$
for each $e'\in B(e)$, since the intersection of $e_\resources$
and $e'_\resources$ is non-empty.
If $|B(e)| \ge 2$, this implies
\begin{multline*}
  z(e_\resources\cup B(e)_\resources) \le \beta / \OPT^* \cdot (v(e_\resources) + v(B(e)_\resources \setminus e_\resources)) \\
  \le \beta \cdot (2/\alpha + |B(e)| \cdot 1/\alpha) \le |B(e)| \cdot \beta \cdot 2 /\alpha = 0.8 \cdot |B(e)| < |B(e)| .
\end{multline*}
Assume in the following that $|B(e)| = 1$.
Let $v_{\min}$ be the value of the smallest element
in $e_\resources\cup B(e)_\resources$.
Then $v(e_\resources\cup B(e)_\resources) \le 2 \OPT^* / \alpha + v_{\min}$: If the smallest element is in $e_\resources$, then
\begin{equation*}
  v(e_\resources\cup B(e)_\resources) \le \underbrace{\OPT^* / \alpha + v_{\min}}_{\ge v(e_\resources)} + v(B(e)_\resources\setminus e_\resources)
  \le 2\OPT^*/\alpha + v_{\min} .
\end{equation*}
The same argument (swapping the role of $e$ and $B(e)$) holds,
if the smallest element is in $B(e)$.
Therefore, if $|B(e)| = 1$ and $v_{\min} \le 1/2 \cdot \OPT^*/\alpha$,
\begin{equation*}
  z(e_\resources\cup B(e)_\resources) \le \beta \cdot (2/\alpha + v_{\min}) \le \beta \cdot 5/2 \cdot 1/\alpha = 1 = |B(e)| .
\end{equation*}
If  $|B(e)| = 1$ and $v_{\min} > 1/2 \cdot \OPT^*/\alpha$,
then $B(e)_\resources\setminus e_\resources, B(e)_\resources\cap e_\resources,$ and $e_\resources\setminus B(e)_\resources$ have at least one and by minimality of edges at most one element.
Since the $z$ value of each thin edge is at most $1/3$,
\begin{equation*}
  z(e_\resources\cup B(e)_\resources) = z(e_\resources\setminus B(e)_\resources) + z(e_\resources\cap B(e)_\resources) + z(B(e)_\resources \setminus e_\resources) \le 3 \cdot 1/3 = 1 = |B(e)| .
\end{equation*}
We conclude that
\begin{equation*}
  \sum_{i\in\players} y_i = |B^f| + |B^t| + 1 > |B^f| + \sum_{e\in A^t} |B(e)| \ge \sum_{j\in\resources^f} z_j + \sum_{e\in A^t} z(e_\resources\cup B(e)_\resources) = \sum_{j\in\resources} z_j . \qedhere
\end{equation*}
\end{proof}

\bibliography{santa2}
\appendix

\end{document}